\documentclass[12pt]{article}

\usepackage[margin=1in]{geometry}
\usepackage{CJKutf8, amsmath, amsthm, authblk, hyperref}

\newtheorem{lemma}{Lemma}
\newtheorem{theorem}{Theorem}

\theoremstyle{definition}
\newtheorem{definition}{Definition}

\theoremstyle{remark}
\newtheorem*{remark}{Remark}

\DeclareMathOperator{\tr}{tr}

\usepackage[style=trad-unsrt, citestyle=numeric-comp, giveninits=true, doi=false, url=false, eprint=false, isbn=false]{biblatex}
\begin{document}

\title{Locally accurate matrix product approximation to thermal states}

\begin{CJK}{UTF8}{gbsn}

\author{Yichen Huang (黄溢辰)\thanks{yichuang@mit.edu}}
\affil{Center for Theoretical Physics, Massachusetts Institute of Technology, Cambridge, Massachusetts 02139, USA}

\maketitle

\end{CJK}

\begin{abstract}

In one-dimensional quantum systems with short-range interactions, a set of leading numerical methods is based on matrix product states, whose bond dimension determines the amount of computational resources required by these methods. We prove that a thermal state at constant inverse temperature $\beta$ has a matrix product representation with bond dimension $e^{\tilde O(\sqrt{\beta\log(1/\epsilon)})}$ such that all local properties are approximated to accuracy $\epsilon$. This justifies the common practice of using a constant bond dimension in the numerical simulation of thermal properties.

\end{abstract}

\section{Introduction}

Classical simulation of quantum many-body systems is a fundamental problem in computational physics. One difficulty is that a generic many-body state cannot even be represented in polynomial space because the dimension of the Hilbert space grows exponentially with the system size. Fortunately, many physically interesting states are non-generic and structured. It may be possible to avoid the ``curse of dimensionality'' by exploiting their structure.

In one-dimensional quantum systems such as spin chains, matrix product state (MPS) methods, including the celebrated density matrix renormalization group, are popular and numerically powerful \cite{Sch11}. As the name suggests, an MPS is a data structure representing a many-body state by products of matrices. The dimension of the matrices is called the bond dimension, which determines the space complexity or the number of parameters in the MPS.

Consider a system in thermal equilibrium. In textbooks and courses, we learned that the system at temperature $T=1/\beta$ is in a mixed state described by the density operator
\begin{equation} \label{eq:th}
\sigma_\beta:=e^{-\beta H}/Z,\quad Z:=\tr e^{-\beta H},  
\end{equation}
where $H$ is the Hamiltonian and $Z$ is the partition function.

Since MPS methods are widely used, it is important to understand their empirical success in simulating the thermal properties of one-dimensional systems with short-range interactions. In a remarkable sequence of papers \cite{Has06, KGK+14, MSVC15, KAA21}, $\sigma_\beta$ is proved to be efficiently approximated by a matrix product operator (MPO) \cite{VGC04, ZV04, GMSC20}, which is a straightforward generalization of MPS to operators. Let $\tilde O(x):=O(x\log x)$. The best known result is
\begin{theorem} [\cite{KAA21}] \label{t:sota}
Consider a chain of $N$ spins governed by a local Hamiltonian $H$. There exists an MPO $\varrho$ with bond dimension
\begin{equation} \label{eq:bd}
    e^{\tilde O\left(\beta^{2/3}+\sqrt{\beta\log\frac{N}\varepsilon}\right)}
\end{equation}
such that $\|\varrho-\sigma_\beta\|_1\le\varepsilon$, where $\|X\|_1:=\tr\sqrt{X^\dag X}$ is the trace norm.
\end{theorem}

\begin{remark}
See Refs. \cite{Has06, KGK+14, MSVC15} for analogues of this theorem in two and higher spatial dimensions.
\end{remark}

In practice, we may not have to increase the bond dimension with the system size $N$. An extreme example is the infinite imaginary time-evolving block decimation algorithm \cite{OV08}, which yields a translationally invariant matrix product representation of $\sigma_\beta$ directly in the thermodynamic limit. It is empirically observed that a constant bond dimension is sufficient for computing expectation values of local observables. This observation cannot be explained by Theorem \ref{t:sota}, where the bond dimension (\ref{eq:bd}) grows with the system size $N$ and diverges in the thermodynamic limit $N\to+\infty$.

In this paper, we prove that there exists an MPO with bond dimension
\begin{equation} \label{eq:main}
    e^{\tilde O\left(\beta^{2/3}+\sqrt{\beta\log\frac1\epsilon}\right)}
\end{equation}
such that all local properties of $\sigma_\beta$ are approximated to accuracy $\epsilon$. For constant $\beta$, the bond dimension (\ref{eq:main})
\begin{itemize}
\item is independent of $N$. This justifies the common practice of using a constant bond dimension in the numerical simulation of thermal properties.
\item grows slower than any power law in $1/\epsilon$, e.g., $\ll1/\epsilon^{0.001}$, as $\epsilon\to0^+$. This explains the empirical observation that high precision can be achieved with a small or moderate bond dimension.
\end{itemize}

As a side remark, Refs. \cite{Hua15aC, SV17, Hua19aA, DB19, Hua19QV} constructed locally accurate matrix product approximations to pure states with low entanglement.

\section{Results}

Consider a chain of $N$ spins with local dimension $d$.

\begin{definition} [matrix product operator]
Let $\{\hat O_j\}_{j=0}^{d^2-1}$ be a basis of the space of linear operators on the Hilbert space of a spin. Let $\{D_i\}_{i=0}^N$ with $D_0=D_N=1$ be a sequence of positive integers. An MPO has the form
\begin{equation} \label{eq:mpo}
\rho=\sum_{j_1,j_2,\ldots,j_N=0}^{d^2-1}\left(A_{j_1}^{(1)}A_{j_2}^{(2)}\cdots A_{j_N}^{(N)}\right)\hat O_{j_1}\otimes\hat O_{j_2}\otimes\cdots\otimes\hat O_{j_N},
\end{equation}
where $A_{j_i}^{(i)}$ is a matrix of size $D_{i-1}\times D_i$. Define $\max_{0\le i\le N}D_i$ as the bond dimension of the MPO $\rho$.
\end{definition}

Consider a local Hamiltonian
\begin{equation}
    H=\sum_{i=1}^{N-1}H_i,\quad\|H_i\|\le1,
\end{equation}
where $H_i$ represents the nearest-neighbor interaction between spins at positions $i,i+1$, and $\|\cdot\|$ is the operator norm. The thermal state $\sigma_\beta$ at inverse temperature $\beta$ is given by Eq. (\ref{eq:th}).

\begin{theorem} \label{t:1D}
There exists an MPO $\rho$ with bond dimension (\ref{eq:main}) such that
\begin{equation}
    |\tr(\rho\hat O)-\tr(\sigma_\beta\hat O)|\le\epsilon
\end{equation}
for any local observable $\hat O$ with $\|\hat O\|\le1$.
\end{theorem}

\begin{proof}
We will purify $\sigma_\beta$. We introduce a second (auxiliary) copy of the system. Spins in the original and auxiliary systems are labeled by $1,2,\ldots,N$ and $\bar1,\bar2,\ldots,\bar N$, respectively. Let $\{|j\rangle\}_{j=0}^{d-1}$ be the computational basis of the Hilbert space of a spin, and
\begin{equation}
    |\Psi\rangle:=\frac{e^{-\beta H/2}\otimes I}{\sqrt Z}\bigotimes_{i=1}^N|\psi\rangle_i,\quad|\psi\rangle_i:=\sum_{j=0}^{d-1}|j\rangle_i|j\rangle_{\bar i},
\end{equation}
where $|\psi\rangle_i$ is an (unnormalized) maximally entangled state of spins $i$ and $\bar i$. By construction, $|\Psi\rangle$ is normalized and is a purification of $\sigma_\beta=\tr_a(|\Psi\rangle\langle\Psi|)$, where $\tr_a$ denotes the partial trace over the auxiliary system. Combining every pair of spins $i,\bar i$ into a composite spin of local dimension $d^2$, we obtain a chain of $N$ composite spins. Let $i|i+1$ be a cut separating the first $i$ and the last $N-i$ composite spins.

\begin{lemma} [Eq. (83) of Ref. \cite{KAA21}] \label{l}
Let $\lambda_1\ge\lambda_2\ge\cdots$ be the Schmidt coefficients of $|\Psi\rangle$ across the cut $i|i+1$ in non-ascending order. Then,
\begin{equation}
    \sum_{j>Q_\delta}\lambda_j^2\le\delta\quad\textnormal{for}\quad Q_\delta:=e^{\tilde O\left(\beta^{2/3}+\sqrt{\beta\log\frac1\delta}\right)}.
\end{equation}
\end{lemma}

Using this lemma and Lemma 4 in Ref. \cite{Hua15aC}, we obtain an MPO $\tilde\rho$ with bond dimension $Q_\delta^2$ such that 
\begin{equation} \label{eq:err}
    \big|\langle\Psi|\hat O|\Psi\rangle-\tr(\tilde\rho\hat O)\big|=O(\sqrt\delta)
\end{equation}
for any local observable $\hat O$ with $\|\hat O\|\le1$. As tracing out the auxiliary system does not increase the bond dimension, $\rho:=\tr_a\tilde\rho$ is an MPO in the original system with bond dimension $Q_\delta^2$. We complete the proof of Theorem \ref{t:1D} by letting $\epsilon$ be the right-hand side of Eq. (\ref{eq:err}).
\end{proof}

\begin{remark}
Recall that $\tilde\rho$ is a locally accurate approximation (\ref{eq:err}) to the purification $|\Psi\rangle$ of $\sigma_\beta$. Since (pure) MPS is more favorable than MPO in both theory and practice, one might prefer $\tilde\rho$ to be an MPS. This can be achieved by using Lemma 4 in Ref. \cite{Hua19aA} instead of Lemma 4 in Ref. \cite{Hua15aC} at the price of weakening the upper bound (\ref{eq:main}) on the bond dimension of $\rho$ to $e^{\tilde O(\beta^{2/3}+\sqrt{\beta\log(1/\epsilon)})}/\epsilon$.
\end{remark}

\section*{Notes}

Very recently, I became aware of related work by Alhambra and Cirac \cite{AC21}, which constructed locally accurate tensor network approximations to thermal states and time evolution in any spatial dimension. Specializing to thermal states in one dimension, my methods and results are different from theirs. Their proof consists of two steps:
\begin{enumerate}
    \item Construct ``local approximations'' assuming exponential decay of correlations.
    \item Merge local approximations using the ``averaging trick'' of Refs. \cite{SV17, DB19, Hua19aA}.
\end{enumerate}
The proof of Theorem \ref{t:1D} uses neither of these ingredients. Different from Eq. (5) in Result 1 of Ref. \cite{AC21}, the bond dimension (\ref{eq:main}) does not depend on the correlation length and grows slower than any power law in $1/\epsilon$ as $\epsilon\to0^+$ for constant $\beta$. This solves an open problem in the conclusion section of Ref. \cite{AC21}.

\section*{Conflict of interest}

The author declares that he has no conflict of interest.

\section*{Acknowledgments}

This work was supported by NSF grant PHY-1818914 and a Samsung Advanced Institute of Technology Global Research Partnership.

\printbibliography

\end{document}